\documentclass[11pt]{article}

\usepackage{url}
\usepackage{pgf}
\usepackage{cite}
\usepackage{titlesec}
\usepackage[english]{babel}
\usepackage{datetime}
\usepackage{mathtools}
\usepackage{geometry}
\usepackage{array}
\usepackage{xspace}
\usepackage{mathrsfs,amsmath,amsfonts,amssymb,amsopn,amsthm}

\usepackage[colorlinks,citecolor=black,linkcolor=black,urlcolor=black]{hyperref}

\makeatletter

\newcommand{\Rmnum}[1]{\expandafter\@slowromancap\romannumeral #1@}
\makeatother

\theoremstyle{plain}

\newtheorem{theorem}{\hspace{1em}Theorem}[section]
\newtheorem{corollary}[theorem]{\hspace{1em}Corollary}
\newtheorem{lemma}[theorem]{\hspace{1em}Lemma}
\newtheorem{remark}[theorem]{\hspace{1em}Remark}
\newtheorem{definition}[theorem]{\hspace{1em}Definition}
\newtheorem{proposition}[theorem]{\hspace{1em}Proposition}
\newtheorem{example}[theorem]{\hspace{1em}Example}

\def\F{\mathbb{F}}

\newcommand{\Tr}{\mathrm{Tr}}
\def\NL{\mathrm{NL}}
\def\nl{\mathrm{nl}}

\def\w{\mathrm{w}}

\allowdisplaybreaks

\begin{document}

\parindent 13pt
\baselineskip 13pt

\title{Weightwise perfectly balanced
functions with high weightwise nonlinearity profile
}

\date{}

\author{Jian Liu\thanks{School of Computer Software, Tianjin University, Tianjin
300350, P. R. China and CNRS, UMR 7539
LAGA, Paris, France, email: jianliu.nk@gmail.com}
\and Sihem Mesnager\thanks{Department of Mathematics, University of Paris VIII, University of Paris XIII, CNRS, UMR 7539
LAGA and Telecom ParisTech, Paris, France, email: smesnager@univ-paris8.fr}
}

\maketitle

\begin{abstract}
Boolean functions with good cryptographic criteria when restricted to the set of vectors with constant Hamming weight play an important role in the recent FLIP stream cipher~\cite{Meaux2016}. In this paper, we propose a large class of weightwise perfectly balanced (WPB) functions, which is not extended affinely (EA) equivalent to the known constructions. We also discuss the weightwise nonlinearity profile of these functions, and present general lower bounds on $k$-weightwise nonlinearity, where $k$ is a power of $2$. Moreover, we exhibit a subclass of the family. By a recursive lower bound, we show that these subclass of WPB functions have very high weightwise nonlinearity profile.
\end{abstract}

{\bf Keywords}: FLIP cipher; Boolean function; weightwise perfectly balance; weightwise nonlinearity

\section{Introduction}\label{Sec: Introduction}

Boolean functions used as primitives in stream ciphers and block ciphers are classically studied with input defined on the whole vector space $\F_2^n$. At Eurocrypt 2016, M\'{e}aux et al.\cite{Meaux2016} proposed a new family of stream ciphers, called FLIP, which is intended for combining with an homomorphic encryption scheme to create an acceptable system of fully
homomorphic encryption.
The symmetric primitive FLIP requires the Hamming weight of the key register to be invariant. This produces a special situation for the structure of filter function: the input of the filter function consists of those vectors in $\F_2^n$ which have constant Hamming weight.
Then, it leads to the problem on how to evaluate the security of a Boolean function with restricted input, i.e., the input of $f$ is a subset of $\F_2^n$.
Besides, in particular stream ciphers, knowing the Hamming weight of a key register enables the attacker to distinguish the keystream from a random bit-stream\cite{JD06}. Therefore, filter functions which have small bias  when restricted to constant weight vectors are preferred.


Early studies on Boolean functions with input restricted to constant weight vectors can be found in~\cite{FilmusCJTCS2016,FilmusEJB2016,FKMW16,FM16}. Their work is asymptotical and from a probability point of view. In 2017, Carlet, M\'{e}aux, and Rotella~\cite{Carlet2016} provided a security analysis on FLIP cipher and gave the first study on cryptographic criteria of Boolean functions with restricted input. An early version of FLIP faces an attack given by Duval et al.~\cite{Duval2016}, which leads the design of the filter function more complicated to reach better criteria on the subsets of $\F_2^n$.
For Boolean functions, the parameters balancedness and nonlinearity are strongly related to the resistance against distinguish attack and affine approximation attack respectively. In~\cite{Carlet2016}, it is shown that, for Boolean functions with restricted input, balancedness and nonlinearity continue to play an important role with respect to the corresponding attacks in the framework of FLIP ciphers.
In particular, Boolean functions which are uniformly distributed over $\{0,1\}$ on $E_{n,k}=\{x\in\F_2^n\mid \mathrm{w_H(x)=k}\}$ for every $1<k<n$ are called \emph{weightwise perfectly balanced} (WPB) functions, where $\mathrm{w_H}(x)$ denotes the Hamming weight of $x$.
The minimum Hamming distance between a Boolean function $f$ and all the affine Boolean functions is called the nonlinearity of $f$. If the input of $f$ is restricted to $E_{n,k}$, then the nonlinearity is called \emph{$k$-weight nonlinearity}. The set of $k$-weight nonlinearity for all $k>1$ is called the \emph{weightwise nonlinearity profile} of  $f$.
The only known construction of WPB functions is due to~\cite{Carlet2016}, which is designed through a recursive method.
Some upper bounds on the $k$-weight nonlinearity of Boolean functions are discussed in~\cite{Carlet2016} and~\cite{MesnagerFq2017} respectively.
As far as we know, there is no known
construction of WPB functions which has high weightwise nonlinearity profile simultaneously.

In this paper, we focus on constructions of WPB functions. We first propose a large family of WPB functions by presenting the trace form as well as the algebraic normal form. Compared with the construction given by Carlet et al.~\cite{Carlet2016}, our family has larger algebraic degree and thus not EA equivalent to the known ones. Then, we discuss the weightwise nonlinearity of these WPB functions, showing that for every $k$ being a positive power of $2$, the $k$-weightwise nonlinearity of any WPB function in our family is nonzero. Furthermore, we construct a subclass of WPB functions in our family, which have high $k$-weightwise nonlinearity for every $k>1$. This is the first time that an infinite class of WPB functions with high weightwise nonlinearity profile has been invented.

The remainder of this paper is organized as follows. Formal definitions and necessary preliminaries
are introduced in Section~\ref{Sec: Preliminaries}. In Section~\ref{Sec: main.con.}, a family of WPB functions is proposed, and the analysis of the weightwise nonlinearity is presented. We exhibit a subclass of WPB functions with high weightwise nonlinearity profile in Section~\ref{Sec: A primary con.}. Finally, we conclude the paper in the last section.


\section{Preliminaries}\label{Sec: Preliminaries}

In this paper, additions and multiple sums calculated modulo $2$ will be denoted by $\oplus$ and $\bigoplus_{i}$ respectively,  additions and multiple sums calculated in characteristic $0$ or in the
additions of elements of the finite field $\F_{2^n}$ will be denoted by $+$ and $\sum_{i}$ respectively.
Let $\F_2^n$ denote the $n$-dimensional vector space over the finite field $\F_2$ with two elements.
An $n$-variable Boolean function $f$ is a function from $\F_2^n$ to $\F_2$.
The $(0,1)$-sequence defined by $(f(\mathbf{v}_0),f(\mathbf{v}_1),\ldots,f(\mathbf{v}_{2^n-1}))$ is called the \emph{truth table} of $f$, where $\mathbf{v}_0=(0,\ldots,0,0), \mathbf{v}_1=(0,\ldots,0,1), \ldots,\mathbf{v}_{2^n-1}=(1,\ldots,1,1)$ are ordered by lexicographical order.
$f$ can be uniquely represented in the \emph{algebraic
normal form} (in brief, ANF) that
\[f(x)=\bigoplus_{v\in\F_2^n}a_{v}x_1^{v_1}x_2^{v_2}\cdots x_n^{v_n},
\]
where
$x=(x_1,\ldots,x_n)\in \F_2^n$, $v=(v_1,\ldots,v_n)\in \F_2^n$,
$a_{v}\in \F_2$.
The \emph{algebraic degree} of $f$, denoted by $\deg(f)$,
is the number of variables in the highest order product term with nonzero coefficient.
A Boolean function is said to be \emph{affine} if $\deg(f)\leqslant 1$.
Two Boolean functions $f$ and $g$ are said to be \emph{extended affinely (EA)}
equivalent if there exist an affine Boolean function $l$ and an affine automorphism $L$ of $\F_2^n$ such that $f=g\circ L \oplus l$.
The algebraic degree of an $n$-variable Boolean function $f$ is affine invariant, i.e., for every affine Boolean function $l$ and every affine automorphism $L$, we have $\deg(f\circ L\oplus l)=\deg(f)$ (see~\cite{Carlet1}).
Given a basis of $\F_{2^n}$ over $\F_2$, $\F_{2^n}$ can be regarded as a vector space over $\F_2$, and there is a bijective $\F_2$-linear mapping from $\F_{2^n}$ to $\F_2^n$. Thus, the field $\F_{2^n}$ can be identified with $\F_2^n$.

Recall that the cyclotomic classes of $2$ modulo $2^n-1$ is defined as $C(j):=\{j2^i \mod (2^n-1)\mid i=0,1,\ldots,o(j)\}$, where $o(j)$ is the smallest positive integer such that $j2^{o(j)}\equiv j~(\mathrm{mod}~(2^n-1))$.
For any positive integers $k$ and $r$ satisfying $r|k$, the trace function from $\F_{2^k}$ to $\F_{2^r}$, denoted by $\Tr_r^k$, is defined as
\[
\Tr_r^k(x):=x+x^{2^r}+x^{2^{2r}}+\cdots+x^{2^{k-r}}, ~~~~x\in\F_{2^k}.
\]
Through the choice of a basis of the vector space $\F_{2^n}$, a Boolean function over $\F_{2^n}$ can be uniquely represented in the following \emph{trace form}~\cite{Carlet2}:
\[
f(x)=\sum_{j\in\Gamma_n}\Tr_1^{o(j)}(a_jx^j)+\epsilon\left(1+x^{2^n-1}\right),
\]
where $\Gamma_n$ is the set of all the coset leaders of the cyclotomic classes of $2$ modulo $2^n-1$,
$o(j)$ is the size of the cyclotomic class of $2$ modulo $2^n-1$ containing $j$,
$a_j\in\F_{2^{o(j)}}$,
$\epsilon=w_{\mathrm{H}}(f) \mod 2$, and $w_{\mathrm{H}}(f)=|\{x\in\F_{2^n} \mid f(x)=1\}|$.
The algebraic degree of $f$ in the above trace form is preserved, which can be read as $\deg(f)=\max\{\mathrm{wt}_2(j), a_j\ne 0\}$ (we make $\epsilon=a_{2^n-1}$), where $\mathrm{wt}_2(j)$ is the number of nonzero coefficients $j_s$ in the binary expansion $\sum_{s=0}^{n-1}j_s 2^s$ of $j$.

Denote by $\w_{\mathrm{H}}(f)_k$ the Hamming weight of a Boolean function $f$ on all the entries with fixed Hamming weight $k$, i.e.,
\[
\w_{\mathrm{H}}(f)_k=\{x\in\F_2^n, \w_{\mathrm{H}}(x)=k, f(x)=1\},
\]
where $\w_{\mathrm{H}}$ denotes the Hamming weight of a vector.

\begin{definition}${}^\text{\!\!\!\!\emph{\cite{Carlet2016}}}$\label{Def.WPB}
     For an $n$-variable Boolean function $f$, $f$ is called \emph{weightwise perfectly balanced (WPB)} if for every $k\in\{1,\ldots,n-1\}$, the restriction of $f$ on $E_{n,k}=\{x\in\F_2^n, \w_{\mathrm{H}}(x)=k\}$ is balanced, i.e., $\w_{\mathrm{H}}(f)_k={n \choose k}/2$.
\end{definition}

It is proved that weightwise perfectly balanced Boolean functions exist only if $n$ is a power of $2$ (see~\cite{Carlet2016}). In this paper, we always consider Boolean functions with $n=2^k$ variables, where $k$ is a positive integer.

\begin{remark}
    For a WPB function $f$, it is imposed that $f(0,\ldots,0)\ne f(1,\ldots,1)$ to make the whole function balanced on $\F_2^n$. Without loss of generality, we suppose that $f(0,\ldots,0)=0$ and $f(1,\ldots,1)=1$.
\end{remark}

Let $E$ be a subset of $\F_{2}^n$ and $f$ be a Boolean function restricted on $E$. The \emph{nonlinearity} of $f$ over $E$, denoted by $\NL_E(f)$, is the minimum Hamming distance between $f$ and all the affine functions restricted to $E$. In particular, the set $\{\NL_{E_{n,k}}(f), k=0,\ldots,n\}$ is called the \emph{weightwise nonlinearity profile} of $f$, where $E_{n,k}=\{x\in\F_2^n, \w_{\mathrm{H}}(x)=k\}$. The value $\NL_{E_{n,k}}(f)$ is  called the \emph{$k$-weight nonlinearity} of $f$, and will be denoted by $\NL_k(f)$ if there is no risk of confusion. The nonlinearity of $f$ over a subset can be calculated as follows.

\begin{proposition}${}^\text{\!\!\!\!\emph{\cite{Carlet2016}}}$\label{prop.restricted nonlineariy}
    Let $f$ be an $n$-variable Boolean function and $E$ be a subset of $\F_2^n$. We have
    \[
    \NL_E(f)=\frac{|E|}{2}-\frac{1}{2}\max_{a\in\F_2^n}\Big|\sum_{x\in E}(-1)^{f(x)\oplus a\cdot x}\Big|,
    \]
    where $a\cdot x$ is the usual inner product defined as $a\cdot x=a_1x_1\oplus\cdots \oplus a_nx_n$ for $a=(a_1,\ldots, a_n)\in\F_2^n$ and $x=(x_1,\ldots,x_n)\in\F_2^n$.
\end{proposition}

In~\cite{Carlet2016}, an upper bound on the nonlinearity of Boolean functions with restricted input is given.
Mesnager~\cite{MesnagerFq2017} presented further advances on this upper bound, and stressed that the improved upper bound might be much lower
than that in~\cite{Carlet2016}.

\begin{proposition}${}^\text{\!\!\!\!\emph{\cite{Carlet2016}}}$\label{prop.upper bound on the restricted nonlineariy}
    Let $f$ be an $n$-variable Boolean function, and $\lfloor a\rfloor$ denote the maximum integer not larger than $a$. Then, for every $E\subseteq F_{2}^n$, we have
    \[
    \NL_E(f)\leqslant \left\lfloor\frac{|E|}{2}-\frac{\sqrt{|E|}}{2}\right\rfloor.
    \]
\end{proposition}

\section{A family of WPB functions}\label{Sec: main.con.}

In this section, we propose a large class of WPB functions, which are not EA equivalent to the functions given by Carlet et al.~\cite{Carlet2016}. For a finite field $\F_{2^n}$, we choose a normal basis $\{\alpha, \alpha^2,\ldots,\alpha^{2^{n-1}}\}$ of $\F_{2^n}$,
and decompose $x\in\F_{2^n}$ over this basis. Thus, if $x=(x_1,\ldots,x_n)$ then $x^2=(x_2,\ldots,x_n,x_1)$, which is a left shift of $x$.
Recall that we always assume $n$ is a power of $2$.

For $0\leqslant s\leqslant n-1$, define the \emph{left $s$-cyclic shift operator} $\rho_n^s$ as
$\rho_n^s(x_i)=x_{(i+s) \mathrm{mod}~n}$, where $x_i\in\F_2$, $1\leqslant i\leqslant n$.
For tuples, we define
$\rho_n^s(x_1,x_2,\ldots, x_n)=\left(\rho_n^s(x_{1}),\rho_n^s(x_{2}),\ldots,\rho_n^s(x_n)\right)$,
and for monomials, we define $\rho_n^s(x_{i_1}x_{i_2}\cdots x_{i_m})=\rho_n^s(x_{i_1}) \rho_n^s(x_{i_2})  \cdots \rho_n^s(x_{i_m})$, where $1\leqslant i_1<i_2<\cdots<i_m\leqslant n$.
An orbit generated by $x$ is defined as $G_k^{(l)}(x)=\{x,\rho_n^1(x),\ldots,\rho_n^{l-1}(x)\}$, where $\mathrm{w_H}(x)=k$ and the length $l$ satisfies $\rho_n^{l}(x)=x$.
Every orbit can be represented by its lexicographically first element, called the \emph{representative element}.
The set of all the representative elements with Hamming weight $k$ and orbit length $l$ is denoted by $\Omega_k^{(l)}$.
For every $E_{n,k}=\{x\in\F_2^n, \w_{\mathrm{H}}(x)=k\}$,
\begin{equation}\label{eq.union of orbits}
    E_{n,k}=\bigcup_{l|n}\bigcup_{x\in\Omega_k^{(l)}}G_k^{(l)}(x).
\end{equation}
Clearly, all the orbits generate a partition of the set $\F_2^n$.
It is proved that (see e.g.~\cite[Appendix A.1]{Daemen1995}) the number of distinct orbits in $\F_2^n$ is $\Psi_n=\frac{1}{n}\sum_{k|n}\phi(k)2^{n/k}$, where $\phi(k)$ is the Euler's \emph{phi}-function.
Define a map $\sigma$ from $\F_{2}^n$ to the set of all the monomials in $\F_{2}^n$ as $\sigma: (x_1,x_2,\ldots,x_n)\in\F_2^n\rightarrow x_{i_1}x_{i_2}\cdots x_{i_m}$, where $x_{i_j}=1$ for $j=1,\ldots, m$, and $x_j=0$ otherwise, $\sigma(0)=1$. It is obvious that $\sigma$ is one-to-one. Then, we also have the concepts of orbit and representative element for monomials, and the number of distinct orbits is $\Psi_n$.

\subsection{General results on the construction of WPB functions}\label{SubSec:main.con.}

\begin{theorem}\label{thm.f=f2 is WPB}
    For a Boolean function $f$ over $\F_{2^n}$, if $f(x^2)=f(x)+1$ holds for all $x\in\F_{2^n}\setminus\{0,1\}$, where $+$ is in $\F_2$, then $f$ is WPB.
\end{theorem}

\begin{proof}
It is easy to see that for an orbit $G^{(l)}(x_0)$ of length $l>1$ in $\F_{2^n}$ (we identify $\F_{2^n}$ with $\F_2^n$ under a normal basis), we have $l|n$. Then, $l$ is power of 2, and thus $l$ is even. Since $f(x^2)=f(x)+1$ for $x\in\F_{2^n}\setminus\{0,1\}$, then $f(x_0)=f(x_0^3)=\cdots=f(x_0^{l-1})=a$ and $f(x_0^2)=f(x_0^4)=\cdots=f(x_0^l)=a\oplus 1$, where $a\in\F_{2}$. Hence, $f$ is balanced on $G^{(l)}(x_0)$. It is clear that all the orbits whose elements have Hamming weight $k$ generate a partition of the set $E_{n,k}=\{x\in\F_2^n, \w_{\mathrm{H}}(x)=k\}$. Therefore, $f$ is balanced on $E_{n,k}$ for $1\leqslant k\leqslant n-1$. According to Definition~\ref{Def.WPB}, $f$ is WPB.
\end{proof}

\begin{theorem}\label{thm.three equivalent forms}
    For a Boolean function $f$ with $f(0)=0$ and $f(1)=1$, the following assertions are equivalent:
    \begin{enumerate}
      \item $f(x^2)=f(x)+1$ holds for all $x\in\F_{2^n}\setminus\{0,1\}$,
      \item $f(x)=\sum_{j\in\Gamma_n\setminus\{0\}}\Tr_1^{o(j)}(\beta^{i_j}x^j)$, where $\beta$ is an primitive element of $\F_{2^2}$, $i_j\in\{1,2\}$ for $j\in\Gamma_n\setminus\{0\}$,
      \item $f(x)=\bigoplus_{i=0}^{l(1)/2-1}\rho_n^{2i+a_1}(x_1)\oplus\bigoplus_{i=0}^{l(1,2)/2-1}\rho_n^{2i+a_2}(x_1x_2)\oplus\cdots\oplus\bigoplus_{i=0}^{l(1,2,\ldots, n-1)/2-1}\rho_n^{2i+a_{\Psi_n-2}}$\\$(x_1x_2\cdots x_{n-1})$, where the monomials in the sums are all the representative elements except for $1$ and $x_1x_2\cdots x_n$, $l(\cdot)$ is the length of the orbit for monomials, and $a_j\in\{0,1\}$ for $j=1,2,\ldots,\Psi_n-2$.
    \end{enumerate}
\end{theorem}

\begin{proof}
We first prove that item 1 is equivalent to item 2. Note that for $j\in\Gamma_n\setminus\{0\}$, $o(j)$ is divisor of $n$, and thus even. Then, for $g(x)=\Tr_1^{o(j)}(\beta x^j)$, we have
\begin{align*}
    g(x)+g(x^2)=&\left(\beta x^j+\beta^2 x^{2j}+\cdots +\beta x^{2^{o(j)-2}j}+\beta^2 x^{2^{o(j)-1}j}\right)\\
                 &+\left(\beta x^{2j}+\beta^2 x^{4j}+\cdots +\beta x^{2^{o(j)-1}j}+\beta^2 x^{j}\right)\\
    =&\left(\beta+\beta^2\right)\Tr_1^{o(j)}(x^j)\\
    =&\Tr_1^{o(j)}(x^j),
\end{align*}
where the last equation is from the fact $\beta+\beta^2=1$ for $\beta$ is an primitive element of $\F_{2^2}$.
Similarly, if $g(x)=\Tr_1^{o(j)}(\beta^2 x^j)$, then $g(x)+g(x^2)=\Tr_1^{o(j)}(x^j)$. Hence, for $f(x)=\sum_{j\in\Gamma_n\setminus\{0\}}\Tr_1^{o(j)}(\beta^{i_j}x^j)$, we have that, for all $x\in\F_{2^n}\setminus\{0,1\}$,
\begin{align*}
    f(x)+f(x^2)=&\sum_{j\in\Gamma_n\setminus\{0\}}\Tr_1^{o(j)}(x^j)=\sum_{j=1}^{2^n-2}x^j=1.
\end{align*}
Thus, item 1 is a necessary condition of item 2.
From the proof of Theorem~\ref{thm.f=f2 is WPB}, we know that the number of Boolean functions satisfying the condition in item 1 is $2^{\Psi_n-2}$. Indeed, the truth table of $f$ in item 1 is determined by the values of $f$ on all the representative elements of the orbits in $\F_{2}^n$, and there are exactly $\Psi_n-2$ distinct orbits in $\F_{2^n}\setminus\{0,1\}$.
On the other hand, it is easy to see that the number of nonzero cyclotomic classes of $2$ modulo $2^n-1$ is $\Psi_n-2$. In fact, every cyclotomic class can be seen as an orbit in $\F_{2}^n$ if the numbers are in binary form, and the orbits $\{(0,0,\ldots ,0)\}$ and $\{(1,1,\ldots, 1)\}$ are corresponding to the cyclotomic class $\{0\}$. Hence, the number of Boolean functions in item 2 is $2^{\Psi_n-2}$. From the above discussion, we get that item 1 is equivalent to item 2.

Now we prove that  item 1 is equivalent to item 3. Recall that $x^2$ is a left shift of $x$. Suppose that $f(x)$ is defined as in item 3, then
\begin{align*}
    f(x)+f(x^2)&=\bigoplus_{i=0}^{l(1)-1}\rho_n^{i}(x_1)\oplus\bigoplus_{i=0}^{l(1,2)-1}\rho_n^{i}(x_1x_2)\oplus\cdots\oplus\bigoplus_{i=0}^{l(1,2,\ldots, n-1)-1}\rho_n^{i}(x_1x_2\cdots x_{n-1})\\
    &=1\oplus x_1x_2\cdots x_n\oplus(x_1\oplus1)(x_2\oplus1)\cdots(x_n\oplus 1).
\end{align*}
 It is easy to see that $1\oplus x_1x_2\cdots x_n\oplus(x_1\oplus1)(x_2\oplus1)\cdots(x_n\oplus 1)=1$ for all $x\in\F_{2^n}\setminus\{0,1\}$, and thus $f(x)$ satisfies the condition in item 1. Moreover, since $a_j\in\{0,1\}$ for $j=1,2,\ldots,\Psi_n-2$, then the number of functions in item 3 is $2^{\Psi_n-2}$ which is equal to the number of functions in item 1. Therefore, item 1 is equivalent to item 3.
\end{proof}

Combining Theorem~\ref{thm.f=f2 is WPB} with Theorem~\ref{thm.three equivalent forms}, we obtain a construction of WPB functions in trace form as well as the algebraic normal form.
Using the trace form, we conclude by the following corollary.

\begin{corollary}\label{corol.trace form}
     The Boolean function $f$ over $\F_{2^n}$, where
     \begin{align}\label{corol.equ.trace form}
        f(x)=\sum_{j\in\Gamma_n\setminus\{0\}}\Tr_1^{o(j)}(\beta^{i_j}x^j)
     \end{align}
     is a WPB function with $\deg(f)=n-1$, where $\beta$ is an primitive element of $\F_{2^2}$, $i_j\in\{1,2\}$ for $j\in\Gamma_n\setminus\{0\}$.
\end{corollary}

\begin{remark}
    From the proof of Theorem~\ref{thm.three equivalent forms}, we know that the number of WPB functions constructed in Corollary~\ref{corol.trace form} is $2^{\Psi_n-2}$.
\end{remark}

\begin{remark}
    In~\cite[Proposition 3]{Carlet2016}, Carlet et al. recursively built a class of WPB functions of $n=2^k$ variables, which has algebraic degree $n/2$. Since the  algebraic degree of the functions in (\ref{corol.equ.trace form}) is $n-1$, we know that the WPB functions in (\ref{corol.equ.trace form}) is not EA equivalent to that in~\cite[Proposition 3]{Carlet2016}. and thus we obtain a new construction of WPB functions.
    Note that from the cryptanalysis viewpoint, the algebraic degree of a Boolean function should be high, but for the Boolean functions used in the filter permutator model (e.g. cipher FLIP). the homomorphic-friendly design requires to reduce the multiplicative depth of the decryption circuit, i.e., a lower algebraic degree is preferred. Thus, there exists a trade off between the security and the performance. 
\end{remark}

\subsection{On the analysis of the weightwise nonlinearity profile of WPB functions}\label{SubSec:main.nonlinearity.}

In this part, we mainly discuss the weightwise nonlinearity profile of the WPB functions given in Corollary~\ref{corol.trace form}.
We first present a property for a normal WPB function.

\begin{proposition}\label{propo. NL1 equals 0}
    For any WPB function $f$, we have $\NL_{1}(f)=0$.
\end{proposition}

\begin{proof}
Let $e_i$ be the identity vector in $\F_2^n$ with $1$ in the $i$-th position and zeros elsewhere.
Since $f$ is balanced on $E_{n,1}=\{x\in\F_2^n, \w_{\mathrm{H}}(x)=1\}$, we have
$f(e_{i_1})=f(e_{i_2})=\cdots=f(e_{i_{n/2}})=1$ and $f(x)=0$ for $x\in E_{n,1}\setminus\{e_{i_1},e_{i_2},\ldots,e_{i_{n/2}}\}$.
Then, it is easy to see that $f$ is equal to the linear function $x_{i_1}\oplus x_{i_2}\oplus\cdots\oplus x_{i_{n/2}}$ when they are restricted to $E_{n,1}$. Thus, $\NL_{1}(f)=0$.
\end{proof}

Let $\Omega_k$ denote the set of all the representative elements with Hamming weight $k$ in $\F_{2}^n$, i.e., $\Omega_k=\bigcup_{l|n}\Omega_{k}^{(l)}$.
For an orbit $G_k^{(l)}(x)$, we denote $\widetilde{G}_k^{(l)}(x)=\{x,\rho_n^2(x),\ldots,\rho_n^{l-2}(x)\}$ (note that since $l|n$, $l$ is even).
Krawtchouk polynomial (see\cite{MacWilliams}) of degree $k$ is defined by $K_k(i,n)=\sum_{j=0}^{k}(-1)^j{i \choose j}{n-i\choose k-j}$. It is known that $\sum_{x\in E_{n,k}}(-1)^{a\cdot x}=K_k(\mathrm{w_H}(a),n)$.

\begin{theorem}
    For a WPB function $f$ in (\ref{corol.equ.trace form}), we have
    \begin{align}\label{eq.K-poly and nonlinearity}
        \NL_k(f)=\frac{1}{2}{n \choose k}-\frac{1}{2}\max_{a\in\bigcup_{k'=1}^{n/2}\Omega_{k'}}\Bigg|K_k(\mathrm{w_H}(a),n)- 2\sum_{x\in\bigcup_{l|n}\bigcup_{\Lambda\in\Omega_k^{(l)}}\widetilde{G}_k^{(l)}(\Lambda)}(-1)^{a\cdot x} \Bigg|,
    \end{align}
    where $2\leqslant k\leqslant n-1$, and $f(\Lambda)=1$ for all $\Lambda\in\Omega_k^{(l)}$.
\end{theorem}

\begin{proof}
According to Proposition~\ref{prop.restricted nonlineariy}, we have
\begin{align*}
\NL_{k}(f) &=   \frac{1}{2}{n \choose k}-\frac{1}{2}\max_{a\in\F_{2}^n}\Bigg|\sum_{x\in E_{n,k}}(-1)^{f(x)\oplus a\cdot x}\Bigg|
\end{align*}
If $\mathrm{w_H}(a)>n/2$, then define $\overline{a}=a+\mathbf{1}$, and thus $0\leqslant\mathrm{w_H}(\overline{a})< n/2$. Since
\begin{align*}
    \Bigg|\sum_{x\in E_{n,k}}(-1)^{f(x)\oplus \overline{a}\cdot x}\Bigg|&=\Bigg|\sum_{x\in E_{n,k}}(-1)^{f(x)\oplus (a+\mathbf{1})\cdot x}\Bigg|\\
    &=\Bigg|\sum_{x\in E_{n,k}}(-1)^{f(x)\oplus a\cdot x\oplus\mathrm{w_H}(x)}\Bigg|\\
    &=\Bigg|(-1)^k\sum_{x\in E_{n,k}}(-1)^{f(x)\oplus a\cdot x}\Bigg|=\Bigg|\sum_{x\in E_{n,k}}(-1)^{f(x)\oplus a\cdot x}\Bigg|,
\end{align*}
and note that $\big|\sum_{x\in E_{n,k}}(-1)^{f(x)}\big|=0$ because $f$ is balanced on $E_{n,k}$, then we have
\begin{align}\label{eq.NL 1<wt(a)<n/2}
\NL_{k}(f)=\frac{1}{2}{n \choose k}-\frac{1}{2}\max_{\scriptstyle a\in\F_{2}^n \atop\scriptstyle 1\leqslant \mathrm{w_H}(a)\leqslant n/2}\Bigg|\sum_{x\in E_{n,k}}(-1)^{f(x)\oplus a\cdot x}\Bigg|.
\end{align}
From (\ref{eq.union of orbits}), we have
\begin{align}
\nonumber    \sum_{x\in E_{n,k}}(-1)^{f(x)\oplus a\cdot x}&=\sum_{x\in\bigcup_{l|n}\bigcup_{\Lambda\in\Omega_k^{(l)}}G_k^{(l)}(\Lambda)}(-1)^{f(x)\oplus a\cdot x}\\
\nonumber    &=\sum_{l|n}\sum_{\Lambda\in\Omega_k^{(l)}}\sum_{x\in G_k^{(l)}(\Lambda)}(-1)^{f(x)\oplus a\cdot x}\\
\label{eq.lambda -rho(lambda)}    &=\sum_{l|n}\sum_{\Lambda\in\Omega_k^{(l)}}\left(\sum_{x\in\widetilde{G}_k^{(l)}(\Lambda)}(-1)^{1\oplus a\cdot x}+\sum_{x\in\widetilde{G}_k^{(l)}(\rho_n^1(\Lambda))}(-1)^{a\cdot x}\right)\\
\nonumber    &=-\sum_{x\in\bigcup_{l|n}\bigcup_{\Lambda\in\Omega_k^{(l)}}\widetilde{G}_k^{(l)}(\Lambda)}(-1)^{a\cdot x}+\sum_{x\in\bigcup_{l|n}\bigcup_{\Lambda\in\Omega_k^{(l)}}\widetilde{G}_k^{(l)}(\rho_n^1(\Lambda))}(-1)^{a\cdot x}\\
\nonumber    &=\sum_{x\in E_{n,k}}(-1)^{a\cdot x}-2\sum_{x\in\bigcup_{l|n}\bigcup_{\Lambda\in\Omega_k^{(l)}}\widetilde{G}_k^{(l)}(\Lambda)}(-1)^{a\cdot x}\\
\label{eq.K poly-2sum}   &=K_k(\mathrm{w_H}(a),n)-2\sum_{x\in\bigcup_{l|n}\bigcup_{\Lambda\in\Omega_k^{(l)}}\widetilde{G}_k^{(l)}(\Lambda)}(-1)^{a\cdot x},
\end{align}
where $f(\Lambda)=1$ for $\Lambda\in\Omega_k^{(l)}$.
Note that $\rho_n^1(a)\cdot x=a\cdot \rho_n^{n-1}(x)$, then from (\ref{eq.lambda -rho(lambda)}), one has
\begin{align*}
    \sum_{x\in E_{n,k}}(-1)^{f(x)\oplus \rho_n^1(a)\cdot x}&=\sum_{l|n}\sum_{\Lambda\in\Omega_k^{(l)}}\left(\sum_{x\in\widetilde{G}_k^{(l)}(\Lambda)}(-1)^{1\oplus a\cdot \rho_n^{n-1}(x)}+\sum_{x\in\widetilde{G}_k^{(l)}(\rho_n^1(\Lambda))}(-1)^{a\cdot \rho_n^{n-1}(x)}\right)\\
    &=\sum_{l|n}\sum_{\Lambda\in\Omega_k^{(l)}}\left(-\sum_{x\in\widetilde{G}_k^{(l)}(\rho_n^1(\Lambda))}(-1)^{1\oplus a\cdot x}+\sum_{x\in\widetilde{G}_k^{(l)}(\Lambda)}(-1)^{a\cdot x}\right)\\
    &=-\sum_{x\in E_{n,k}}(-1)^{f(x)\oplus a\cdot x}.
\end{align*}
Then, due to (\ref{eq.NL 1<wt(a)<n/2}), we know that
\begin{align}\label{eq.NL a in Omega}
\NL_{k}(f)=\frac{1}{2}{n \choose k}-\frac{1}{2}\max_{a\in\bigcup_{k'=1}^{n/2}\Omega_{k'}}\Bigg|\sum_{x\in E_{n,k}}(-1)^{f(x)\oplus a\cdot x}\Bigg|.
\end{align}
Combining (\ref{eq.K poly-2sum}) with (\ref{eq.NL a in Omega}), we obtain the desired result.
\end{proof}

We now focus on general lower bounds on the $k$-weight nonlinearity of WPB functions in (\ref{corol.equ.trace form}).
Let $\NL_k^{(n)}$ denote the lower bound on $k$-weight nonlinearity for all WPB functions over $\F_{2^n}$ in (\ref{corol.equ.trace form}), i.e., for any WPB function $f$ over $\F_{2^n}$ in (\ref{corol.equ.trace form}), $\NL_k(f)\geqslant \NL_k^{(n)}$. Then, we have the following result.

\begin{theorem}\label{thm.NL(n-k)=NL(k)}
    For $1\leqslant k\leqslant n/2$, $\NL_{n-k}^{(n)}=\NL_{k}^{(n)}$.
\end{theorem}

\begin{proof}
 It is clear that $E_{n,n-k}=E_{n,k}+1=\{x+\mathbf{1}\mid x\in\F_2^n, \mathrm{w_H}(x)=k\}$, where $x+\mathbf{1}=(x_1\oplus1,\ldots,x_n\oplus 1)$. Then, for any  WPB function $f$ in (\ref{corol.equ.trace form}), there exists a WPB function $g$ in (\ref{corol.equ.trace form}) such that $f(x)=g(x+\mathbf{1})$ for any $x\in E_{n,k}$. Hence, $\NL_k(f)=\NL_{n-k}(g)\geqslant \NL_{n-k}^{(n)}$, and thus $\NL_k^{(n)}\geqslant \NL_{n-k}^{(n)}$. Conversely, there exists another  WPB function $h$ in (\ref{corol.equ.trace form}) such that $f(x)=h(x+\mathbf{1})$ for any $x\in E_{n,n-k}$. Hence, $\NL_{n-k}(f)=\NL_{k}(h)\geqslant \NL_{k}^{(n)}$, and thus $\NL_{n-k}^{(n)}\geqslant \NL_k^{(n)}$. Therefore, we obtain $\NL_{n-k}^{(n)}=\NL_{k}^{(n)}$.
\end{proof}

\begin{remark}
    Because of Proposition~\ref{propo. NL1 equals 0} and Theorem\ref{thm.NL(n-k)=NL(k)}, we only need to consider $\NL_{k}^{(n)}$, where $2\leqslant k\leqslant n/2$.
\end{remark}

\begin{example}\label{example.n=8}
   In Table~\ref{weightwise non. pro. n=8}, we calculate the weightwise nonlinearity profile for all $f$ in (\ref{corol.equ.trace form}) with $n=8$ variables by MAGMA.
   Due to Proposition~\ref{propo. NL1 equals 0} and the proof of Theorem~\ref{thm.NL(n-k)=NL(k)}, we only need to consider $\NL_k(f)$ for $k=2,3,4$.
   It is shown that for the best case, the $k$-weight nonlinearity of $f$ is near the upper bound in Proposition~\ref{prop.upper bound on the restricted nonlineariy}. In particular,  if $f$ satisfies
   \begin{align}\label{example.constraint on f with n=8}
    f(0,1,1,1,0,0,0,0)\ne f(1,1,0,1,0,0,0,0),
   \end{align}
 then $\NL_3(f)\geqslant 8$. Note that $(0,1,1,1,0,0,0,0)$ and $(1,1,0,1,0,0,0,0)$ are in different orbits.

\begin{table}[t!]
\caption{\label{weightwise non. pro. n=8} Weightwise nonlinearity profile of WPB functions in (\ref{corol.equ.trace form})  with $n=8$ variables }
   \begin{center}
    \begin{tabular}{|c|c|}%
     \hline
     $k$-weight nonlinearity of $f$ & $\left\lfloor {n \choose k}/2-\sqrt{ {n \choose k}}/2\right\rfloor$ \\  \hline
     $\NL_2(f)\in\{6,9\}$ & $11$\\
      $\NL_3(f)\in\{0,8,14,16,18,20,21,22\}$ & $24$ \\
      $\NL_4(f)\in\{19,22,23,24,25,26,27\}$ & $30$ \\
     \hline
   \end{tabular}
   \end{center}
 \end{table}
\end{example}

\begin{theorem}
    For any $n=2^l\geqslant 8$, we have
      \begin{align*}
        \NL_{2^i}^{(n)}\geqslant \left\{\begin{array}{ll}
        5,& \mbox{if}~2\leqslant i< l-1,\\
        19,& \mbox{if}~i=l-1.\end{array}\right.
      \end{align*}
\end{theorem}

\begin{proof}
We first prove that for any $n\geqslant 8$, $\NL_2^{(n)}\geqslant 5$. Note that for a Boolean function $f$, it is clear that $\NL_E(f)\geqslant \NL_S(f)$ if $S\subseteq E$. Let
    \begin{align*}
        S=\{&(1,1,0,0,0,0,0,0,0,\ldots,0),(0,1,1,0,0,0,0,0,0,\ldots,0),\\
            &(0,0,1,1,0,0,0,0,0,\ldots,0),(0,0,0,1,1,0,0,0,0,\ldots,0),\\
            &(0,0,0,0,1,1,0,0,0,\ldots,0),(0,0,0,0,0,1,1,0,0,\ldots,0),\\
            &(0,0,0,0,0,0,1,1,0,\ldots,0),\\
            &(1,0,1,0,0,0,0,0,0,\ldots,0),(0,1,0,1,0,0,0,0,0,\ldots,0),\\
            &(0,0,1,0,1,0,0,0,0,\ldots,0),(0,0,0,1,0,1,0,0,0,\ldots,0),\\
            &(0,0,0,0,1,0,1,0,0,\ldots,0),(0,0,0,0,0,1,0,1,0,\ldots,0),\\
            &(1,0,0,1,0,0,0,0,0,\ldots,0),(0,1,0,0,1,0,0,0,0,\ldots,0),\\
            &(0,0,1,0,0,1,0,0,0,\ldots,0),(0,0,0,1,0,0,1,0,0,\ldots,0),\\
            &(0,0,0,0,1,0,0,1,0,\ldots,0),\\
            &(1,0,0,0,1,0,0,0,0,\ldots,0),(0,1,0,0,0,1,0,0,0,\ldots,0),\\
            &(0,0,1,0,0,0,1,0,0,\ldots,0),(0,0,0,1,0,0,0,1,0,\ldots,0)\}
    \end{align*}
    It is obvious that $S$ is a subset of $E_{n,2}$ for any $n\geqslant 8$, and there exist $4$ subsets of distinct orbits respectively in $S$.  Since for any $x\in S$, $x=(x_1,\ldots,x_8,0,\ldots,0)$, then we only need to consider $a\cdot x$ with $a=(a_1,\ldots,a_8,0,\ldots,0)$ as linear functions on $S$. By some computations, we find that for any $f$ in (\ref{corol.equ.trace form}), $\NL_S(f)\geqslant 5$, and thus $\NL_2(f)\geqslant 5$. Hence, $\NL_2^{(n)}\geqslant 5$.

    For any $n=2^l\geqslant 8$ and any $2\leqslant i\leqslant l-1$, let $$R=\Big\{(y,y)\mid y\in\F_2^{n/2}, \mathrm{w_H}(y)=2^{i-1}\Big\}.$$ Clearly, for any $x=(y,y)\in R$, the orbit generated by $x$ satisfies $G_{2^i}^{(l)}(x)=G_{2^{i-1}}^{(l)}(y)$. Then, for any WPB function $f$ over $\F_{2^n}$ in (\ref{corol.equ.trace form}), there must exist a WPB function $g$ over $\F_{2^{n/2}}$ in (\ref{corol.equ.trace form}) such that $f(x)=g(y)$, where $x=(y,y)\in R$. Since $R\subseteq E_{n,2^i}$ and $E_{n/2,2^{i-1}}=\{y\in\F_2^{n/2}, \mathrm{w_H}(y)=2^{i-1}\}$, then $\NL_{2^i}(f)\geqslant \NL_R(f)=\NL_{2^{i-1}}(g)$, which leads to
    \begin{align*}
        \NL_{2^{i}}^{(n)}\geqslant \NL_{2^{i-1}}^{(n/2)}\geqslant \cdots \geqslant
        \left\{\begin{array}{ll}
        \NL_{2}^{(n/2^{i-1})}\geqslant 5, & \mbox{if} ~2\leqslant i< l-1,\\
        \NL_{4}^{(8)}\geqslant 19, & \mbox{if}~ i=l-1,
        \end{array}\right.
    \end{align*}
    where $\NL_{4}^{(8)}\geqslant 19$ is according to Table~\ref{weightwise non. pro. n=8}.
\end{proof}


\section{A primary construction of WPB functions with high weightwise nonlinearity profile}\label{Sec: A primary con.}

In this section, we propose Construction-1 as a subclass of WPB functions in (\ref{corol.equ.trace form}), and then we prove that these WPB functions have high weightwise nonlinearity profile.
From Theorem~\ref{thm.three equivalent forms} and the proof of Theorem~\ref{thm.f=f2 is WPB}, we know that to obtain a WPB function $f$ in (\ref{corol.equ.trace form}), one only needs to define the values of $f$ on all the representative elements of the orbits in $\F_{2}^n$.
So, in Construction-1 below, we only define $f$ on the representative elements of the orbits in $\F_{2}^n$.
Recall that $\Omega_k$ denotes the set of all the representative elements with Hamming weight $k$ in $\F_{2}^n$.
By Lemma~\ref{lem.Con-1 is reasonable}, we give more explanations on Construction-1.


\begin{table}[t!]
\begin{center}
{\noindent\begin{tabular}{|l|}
  \hline
  Construction-1.\\
  \hline
  Input: Parameter $n=2^l$, $l\geqslant 3$.\\
  Output: An $n$-variable WPB function $f$.\\
  \hline
\textbf{1.} If $n=8$, then output any function in (\ref{corol.equ.trace form}) with constraint in (\ref{example.constraint on f with n=8}).\\
\textbf{2.} If $n\geqslant 16$, let $f(0)=0$, $f(1)=1$. Define $Y_k=\Big\{y\in\F_{2}^{n/4-1}\mid \mathrm{w_H}(y)=\lceil k/2\rceil-2\Big\}$ for\\
~~~  $3\leqslant k\leqslant n-1$, then\\
      ~~~\textbf{2.1.} Suppose $k=2$. For $x\in \Omega_2$, $f(x)$ is chosen randomly in $\F_2$.\\
      ~~~\textbf{2.2.} Suppose $k=2i-1$, $2\leqslant i\leqslant n/2$. For $y_1\in Y_k$, let\\
      ~~~~~~~~~~~~~~~$R_{y_1}=\Big\{(1,y_1,\mathbf{0},y_2)\mid  y_2\in\F_2^{n/2},\mathrm{w_H}(y_2)=i, \mathbf{0}\in\F_{2}^{n/4}\Big\}$,\\
      ~~~~~~~~~~~and for $x\in R_{y_1}$, define $f(x)=g(y_2)$, where $g$ is an $n/2$-variable function in\\
      ~~~~~~~~~~~Construction-1. For $x\in \Omega_k\setminus\bigcup_{y_1\in Y_k} R_{y_1}$, $f(x)$ is chosen randomly in $\F_2$.\\
      ~~~\textbf{2.3.} Suppose $k=2i$, $2\leqslant i\leqslant n/2$. For $y_1\in Y_k$, let\\
      ~~~~~~~~~~$T_{y_1}=\Big\{(1,y_1,\mathbf{0},y_2)\mid  y_2\in\F_2^{n/2},\mathrm{w_H}(y_2)=i+1, \mathbf{0}\in\F_{2}^{n/4}\Big\}$,\\
      ~~~~~~~~~~$S_{y_1}=\Big\{(1,y_1,\mathbf{0},1,y_2)\mid y_2\in\F_2^{n/2},\mathrm{w_H}(y_2)=i, y_2\ne(y,\mathbf{0},1,1), y\in Y_k, \mathbf{0}\in\F_{2}^{n/4-1}\Big\}$,\\
      ~~~~~~~~~~and for $x_1=(1,y_1,\mathbf{0},y_2)\in T_{y_1}$, $x_2=(1,y_1,\mathbf{0},1,z_2)\in S_{y_1}$, define $f(x_1)=g_1(y_2)$, \\ ~~~~~~~~~~$f(x_2)=g_2(z_2)$, where $g_1$, $g_2$
      are $n/2$-variable functions in Construction-1. For\\
      ~~~~~~~~~~$x\in \Omega_k\setminus \bigcup_{y_1\in Y_k}(T_{y_1}\bigcup S_{y_1})$, $f(x)$ is chosen randomly in $\F_2$.\\
  \hline
\end{tabular}}
\end{center}
\end{table}


\begin{lemma}\label{lem.Con-1 is reasonable}
    Construction-1 outputs an $n$-variable WPB function.
\end{lemma}

\begin{proof}
    We prove that for any $k\geqslant 3$, $\bigcup_{y_1\in Y_k}R_{y_1}$ and $\bigcup_{y_1\in Y_k}(T_{y_1}\bigcup S_{y_1})$ consist of distinct orbits in $\F_{2}^n$. Thus, by Construction-1, we can define a WPB function which has form (\ref{corol.equ.trace form}).

For $\bigcup_{y_1\in Y_k}R_{y_1}$, suppose that there exists some $j\geqslant 1$ such that $x_1=\rho_n^j(x_2)$, where $x_1,x_2\in \bigcup_{y_1\in Y_k}R_{y_1}$. Then, since the first coordinate of $x_1$ is $1$, then it must be the case that $\rho_n^j(x_2)=(y_1'',\mathbf{0}, y_2,1,y_1')$ or $\rho_n^j(x_2)=(y_2'',1,y_1,\mathbf{0},y_2')$, where $x_2=(1,y_1,\mathbf{0},y_2)$, $\mathbf{0}\in\F_2^{n/4}$, $y_1'\|y_1''=y_1\in\F_2^{n/4-1}$, $y_2'\|y_2''=y_2\in\F_2^{n/2}$, $\mathrm{w_H}(y_1)=i-2$, $\mathrm{w_H}(y_2)=i$, and $\|$ means the concatenation of two vectors.
\begin{itemize}
  \item Suppose $\rho_n^j(x_2)=(y_1'',\mathbf{0}, y_2,1,y_1')=(1,z_1,\mathbf{0},z_2)=x_1\in R_{z_1}$. Since $\mathrm{w_H}(y_1'')\leqslant i-2$, and the first coordinate of $y_1''$ is 1, then $\mathrm{w_H}(z_1)\leqslant i-3$, which contradicts with $\mathrm{w_H}(z_1)=i-2$.

  \item Suppose $\rho_n^j(x_2)=(y_2'',1,y_1,\mathbf{0},y_2')=(1,z_1,\mathbf{0},z_2)=x_1\in R_{z_1}$.
  Since $\mathrm{w_H}(1,y_1)=i-1$, and the first coordinate of $y_2''$ is 1, then $\mathrm{w_H}(z_1)\geqslant i-1$, which contradicts with $\mathrm{w_H}(z_1)=i-2$.
\end{itemize}
Therefore, all the elements in $\bigcup_{y_1\in Y_k}R_{y_1}$ belong to different orbits in $\F_{2}^n$.

For $\bigcup_{y_1\in Y_k}(T_{y_1}\bigcup S_{y_1})$, since we can prove similarly that all the elements in $\bigcup_{y_1\in Y_k}T_{y_1}$ belong to different orbits in $\F_{2}^n$, then we only consider the following two cases.

1. Suppose that there exists some $j\geqslant 1$ such that $x_1=\rho_n^j(x_2)$, where $x_1,x_2\in \bigcup_{y_1\in Y_k}S_{y_1}$. Then, since the first coordinate of $x_1$ is $1$, then it must be the case that $\rho_n^j(x_2)=(1,y_2,1,y_1,\mathbf{0})$, $\rho_n^j(x_2)=(y_1'',\mathbf{0}, 1, y_2,1,y_1')$, or $\rho_n^j(x_2)=(y_2'',1,y_1,\mathbf{0},1,y_2')$, where $x_2=(1,y_1,\mathbf{0},1,y_2)$, $\mathbf{0}\in\F_2^{n/4-1}$, $y_1'\|y_1''=y_1\in\F_2^{n/4-1}$, $y_2'\|y_2''=y_2\in\F_2^{n/2}$, $\mathrm{w_H}(y_1)=i-2$, $\mathrm{w_H}(y_2)=i$.
\begin{itemize}
  \item Suppose $\rho_n^j(x_2)=(1,y_2,1,y_1,\mathbf{0})=(1,z_1,\mathbf{0},1,z_2)=x_1\in S_{z_1}$. Let $y_2=(b_1,b_2)\in\F_2^{n/2-1}\times\F_2$. Since $\mathrm{w_H}(1,b_1)=\mathrm{w_H}(1,z_1,\mathbf{0},1)=i$ and $\mathrm{w_H}(y_2)=i$, then $b_2=1$, and thus $y_2=(z_1,\mathbf{0},1,1)$, which contradicts with the condition $y_2\ne (y,\mathbf{0},1,1)$ for $y\in Y_k$.
  \item Suppose $\rho_n^j(x_2)=(y_1'',\mathbf{0}, 1, y_2,1,y_1')=(1,z_1,\mathbf{0},1,z_2)=x_1\in S_{z_1}$. Since $\mathrm{w_H}(y_1'')\leqslant i-2$, and the first coordinate of $y_1''$ is 1, then $\mathrm{w_H}(y_1)\leqslant i-3$, which contradicts with $\mathrm{w_H}(y_1)=i-2$.

  \item Suppose $\rho_n^j(x_2)=(y_2'',1,y_1,\mathbf{0},1,y_2')=(1,z_1,\mathbf{0},1,z_2)=x_1\in S_{z_1}$.
  Since $\mathrm{w_H}(1,y_1)=i-1$, and the first coordinate of $y_2''$ is 1, then $\mathrm{w_H}(y_1)\geqslant i-1$, which contradicts with $\mathrm{w_H}(y_1)=i-2$.
\end{itemize}
Therefore, all the elements in $\bigcup_{y_1\in Y_k}S_{y_1}$ belong to different orbits in $\F_{2}^n$.

2. Suppose that there exists some $j\geqslant 1$ such that $x_1=\rho_n^j(x_2)$, where $x_1\in \bigcup_{y_1\in Y_k}T_{y_1}$, $x_2\in \bigcup_{y_1\in Y_k}S_{y_1}$. Then, since the first coordinate of $x_1$ is $1$, then it must be the case that $\rho_n^j(x_2)=(1,y_2,1,y_1,\mathbf{0}_2)$, $\rho_n^j(x_2)=(y_1'',\mathbf{0}_2, 1, y_2,1,y_1')$, or $\rho_n^j(x_2)=(y_2'',1,y_1,\mathbf{0}_2,1,y_2')$, where $x_2=(1,y_1,\mathbf{0},1,y_2)$, $\mathbf{0}_2\in\F_2^{n/4-1}$, $y_1'\|y_1''=y_1\in\F_2^{n/4-1}$, $y_2'\|y_2''=y_2\in\F_2^{n/2}$, $\mathrm{w_H}(y_1)=i-2$, $\mathrm{w_H}(y_2)=i$.
\begin{itemize}
  \item Suppose $\rho_n^j(x_2)=(1,y_2,1,y_1,\mathbf{0}_2)=(1,z_1,\mathbf{0}_1,1,z_2)=x_1\in T_{z_1}$, where $\mathbf{0}_1\in\F_2^{n/4}$. Let $y_2=(b_1,b_2)\in\F_2^{n/2-1}\times\F_2$, then $(1,b_1)=(1,z_1,\mathbf{0}_1)$.
      Since $\mathrm{w_H}(y_2)=i$, then $\mathrm{w_H}(1,b_1)\geqslant i$, which contradicts with $\mathrm{w_H}(1,z_1,\mathbf{0}_1)=i-1$.
  \item Suppose $\rho_n^j(x_2)=(y_1'',\mathbf{0}_2, 1, y_2,1,y_1')=(1,z_1,\mathbf{0},1,z_2)=x_1\in T_{z_1}$, where $\mathbf{0}_1\in\F_2^{n/4}$. Since $\mathrm{w_H}(y_1'')\leqslant i-2$, and the first coordinate of $y_1''$ is 1, then $\mathrm{w_H}(z_1)\leqslant i-3$, which contradicts with $\mathrm{w_H}(z_1)=i-2$.

  \item Suppose $\rho_n^j(x_2)=(y_2'',1,y_1,\mathbf{0},1,y_2')=(1,y_1,\mathbf{0},1,z_2)=x_1\in T_{z_1}$, where $\mathbf{0}_1\in\F_2^{n/4}$.
  Since $\mathrm{w_H}(1,y_1)=i-1$, and the first coordinate of $y_2''$ is 1, then $\mathrm{w_H}(y_1)\geqslant i-1$, which contradicts with $\mathrm{w_H}(y_1)=i-2$.
\end{itemize}
Therefore, $\bigcup_{y_1\in Y_k}T_{y_1}$ and $\bigcup_{y_1\in Y_k}S_{y_1}$ consist of different orbits in $\F_{2}^n$.
In conclusion, for any $k\geqslant 3$,  $\bigcup_{y_1\in Y_k}R_{y_1}$ and $\bigcup_{y_1\in Y_k}(T_{y_1}\bigcup S_{y_1})$ consist of distinct orbits in $\F_{2}^n$.
\end{proof}

We now discuss about the lower bounds on weightwise nonlinearity profile of WPB functions in Construction-1. Let $\nl_k^{(n)}$ denote the lower bound on $k$-weight nonlinearity for all $n$-variable WPB functions in Construction-1, i.e., for any $n$-variable WPB function $f$ in Construction-1, $\nl_k(f)\geqslant \nl_k^{(n)}$.
Similar to the proof of Theorem~\ref{thm.NL(n-k)=NL(k)}, we can obtain $\nl_k^{(n)}=\nl_{n-k}^{(n)}$. Hence, in the following, one only needs to consider $\nl_k^{(n)}$ for $2\leqslant k\leqslant n/2$.

Before present the lower bound in Theorem~\ref{thm.lower bound on NL recursively}, we first see the following lemma.

\begin{lemma}\label{lem.NL(f(E))=NL(f(rho(E)))}
    For a set $E\subseteq \F_2^n$ and $0\leqslant j\leqslant n-1$, define $\rho_n^j(E)=\{\rho_n^j(x)\mid x\in E\}$. Then, if an $n$-variable Boolean function $f$ satisfies $f(x)=f(\rho_n^j(x))$ for all $x\in E$, then $\NL_{\rho_n^j(E)}(f)=\NL_E(f)$.
\end{lemma}

\begin{proof}
    For any $a,x\in\F_2^n$ and $0\leqslant j\leqslant n-1$, we have $a\cdot \rho_n^j(x)=\rho_n^{n-j}(a)\cdot x$.
    Since $f(x)=f(\rho_n^j(x))$ for all $x\in E$, then
    \begin{align}\label{lem.eq.|p(E)|=|E|}
        \sum_{x\in \rho_n^j(E)}(-1)^{f(x)\oplus a\cdot x}=\sum_{x\in E}(-1)^{f\left( \rho_n^j(x)\right)\oplus a\cdot  \rho_n^j(x)}
    =\sum_{x\in E}(-1)^{f(x)\oplus \rho_n^{n-j}(a)\cdot x}.
    \end{align}
    According to (\ref{lem.eq.|p(E)|=|E|}), we have
    \begin{align*}
        \NL_{\rho_n^j(E)}(f)&=\frac{1}{2}|\rho_n^j(E)|-\frac{1}{2}\max_{a\in\F_2^n}\Bigg|\sum_{x\in \rho_n^j(E)}(-1)^{f(x)\oplus a\cdot x}\Bigg|\\
        &=\frac{1}{2}|E|-\frac{1}{2}\max_{a\in\F_2^n}\bigg|\sum_{x\in E}(-1)^{f(x)\oplus a\cdot x}\bigg|=\NL_E(f).
    \end{align*}

\end{proof}

\begin{theorem}\label{thm.lower bound on NL recursively}
    For $n\geqslant 8$ and $2\leqslant i\leqslant n/4$, we have the following lower bound on weighwise nonlinearity profile recursively,
    \begin{align*}
        \nl_2^{(n)}&\geqslant 5,\\
        \nl_{2i-1}^{(n)}&\geqslant n{n/4-1\choose i-2}\nl_i^{(n/2)},\\
        \nl_{2i}^{(n)}&\geqslant \frac{n}{2}{n/4-1\choose i-2}\left(2\nl_i^{(n/2)}-2{n/4-1\choose i-2}-1\right)+n{n/4-1\choose i-2}\nl_{i+1}^{(n/2)}.
    \end{align*}
\end{theorem}

\begin{proof}
For any $n$-variable WPB function $f$ in Construction-1, We first consider the case $k=2i-1$ for some $2\leqslant i\leqslant n/4$.
Since for any $y_1\in Y_k$, $f(x)=g(y_2$) for all $x=(1,y_1,\mathbf{0},y_2)\in R_{y_1}$. So, we have
\begin{align}\label{thm.bound on nl, eq.1}
\NL_{R_{y_1}}(f)=\NL_i(g)\geqslant \nl_i^{(n/2)}.
\end{align}
From Lemma~\ref{lem.NL(f(E))=NL(f(rho(E)))}, we have $\NL_{\rho_n^j(R_{y_1})}(f)=\NL_{R_{y_1}}(f)$ for any $0\leqslant j\leqslant n-1$.
Because $k$ is odd, it is easy to prove that for any $x\in R_{y_1}$, the orbit generated by $x$ has length $n$.
Then, $|\bigcup_{j=0}^{n-1}\rho_n^j(R_{y_1})|=n|R_{y_1}|$, and thus (\ref{thm.bound on nl, eq.1}) leads to
\begin{align}\label{thm.bound on nl, eq.2}
    \NL_{\bigcup_{j=0}^{n-1}\rho_n^j(R_{y_1})}(f)\geqslant n\cdot \NL_{\rho_n^j(R_{y_1})}(f)=n\cdot \NL_{R_{y_1}}(f)\geqslant n\cdot \nl_i^{(n/2)}.
\end{align}
Let $Y=\bigcup_{y_1 \in Y_k}\bigcup_{j=0}^{n-1}\rho_n^j(R_{y_1})$, then according to (\ref{thm.bound on nl, eq.2}), we obtain
\begin{align*}
    \NL_Y(f)\geqslant |Y_k|\cdot \NL_{\bigcup_{j=0}^{n-1}\rho_n^j(R_{y_1})}(f)\geqslant n{n/4-1\choose i-2}\nl_i^{(n/2)}.
\end{align*}
Note that $Y\subseteq E_{n,k}$. Then, $\NL_k(f)\geqslant \NL_Y(f)$, and thus $\nl_{2i-1}^{(n)}\geqslant n{n/4-1\choose i-2}\nl_i^{(n/2)}$.

Let $k=2i$ with $2\leqslant i\leqslant n/4$.
Since for any $y_1\in Y_k$,
$f(x_1)=g_1(y_2)$ and $f(x_2)=g_2(z_2)$ for all $x_1=(1,y_1,\mathbf{0},y_2)\in T_{y_1}$ and $x_2=(1,y_1,\mathbf{0},1,z_2)\in S_{y_1}$.
So, we have
\begin{align}
\label{thm.bound on nl, eq.3}\NL_{T_{y_1}}(f)&=\NL_{i+1}(g_1)\geqslant \nl_{i+1}^{(n/2)},\\
\label{thm.bound on nl, eq.4}\NL_{S_{y_1}}(f)&=\NL_i(g_2)-|Y_k|\geqslant \nl_i^{(n/2)}-{n/4-1\choose i-2},
\end{align}
where (\ref{thm.bound on nl, eq.4}) is because $z_2\ne (y,\mathbf{0},1,1)$ for all $y\in Y_k$.
It is easy to prove that for any $x=(1,y_1,\mathbf{0},y_2)\in T_{y_1}$, the orbit generated by $x$ has length $n$, and for any $x=(1,y_1,\mathbf{0},1,z_2)\in S_{y_1}$, the orbit generated by $x$ has length $n/2$ if $z_2=(1,y_1,\mathbf{0},1)$, and $n$ otherwise.
Then, $|\bigcup_{j=0}^{n-1}\rho_n^j(T_{y_1})|=n\cdot|T_{y_1}|$, $|\bigcup_{j=0}^{n/2-1}\rho_n^j(S_{y_1})|=n/2\cdot|S_{y_1}|$, and $|\bigcup_{j=n/2}^{n-1}\rho_n^j(S_{y_1}\setminus e)|=n/2\cdot|S_{y_1}-1|$, where $e=(1,y_1,\mathbf{0},1,1,y_1,\mathbf{0},1)$. Define
\begin{align*}
    T&=\bigcup_{y_1\in Y_k}\bigcup_{j=0}^{n-1}\rho_n^j(T_{y_1}),\\
    S_1&=\bigcup_{y_1\in Y_k}\bigcup_{j=0}^{n/2-1}\rho_n^j(S_{y_1}),\\
    S_2&=\bigcup_{y_1\in Y_k}\bigcup_{j=n/2}^{n-1}\rho_n^j(S_{y_1}\setminus\{(1,y_1,\mathbf{0},1,1,y_1,\mathbf{0},1)\}).
\end{align*}
Then, we have
\begin{align}
\nonumber    \NL_{T\bigcup S_1\bigcup S_2}(f)\geqslant &\NL_{T}(f)+\NL_{S_1}(f)+\NL_{S_2}(f)\\
\nonumber    =&n\cdot |Y_k| \cdot \NL_{\rho_n^j(T_{y_1})}(f)+\frac{n}{2}\cdot|Y_k|\cdot\NL_{\rho_n^j(S_{y_1})}(f)\\
\nonumber &+\frac{n}{2}\cdot|Y_k|\cdot\left(\NL_{\rho_n^j(S_{y_1})}(f)-1\right)\\
\label{thm.bound on nl, eq.5}    =&n\cdot |Y_k| \cdot \NL_{T_{y_1}}(f)+\frac{n}{2}\cdot|Y_k|\cdot\NL_{S_{y_1}}(f)+\frac{n}{2}\cdot|Y_k|\cdot\left(\NL_{S_{y_1}}(f)-1\right)\\
 \nonumber   \geqslant& n{n/4-1\choose i-2}\nl_{i+1}^{(n/2)}+\frac{n}{2}{n/4-1\choose i-2}\left(\nl_i^{(n/2)}-{n/4-1\choose i-2}\right)\\
\label{thm.bound on nl, eq.6}&+\frac{n}{2}{n/4-1\choose i-2}\left(\nl_i^{(n/2)}-{n/4-1\choose i-2}-1\right)\\
 \nonumber =&  n{n/4-1\choose i-2}\nl_{i+1}^{(n/2)}+\frac{n}{2}{n/4-1\choose i-2}\left(2\nl_i^{(n/2)}-2{n/4-1\choose i-2}-1\right),
\end{align}
where (\ref{thm.bound on nl, eq.5}) is due to Lemma~\ref{lem.NL(f(E))=NL(f(rho(E)))}, and (\ref{thm.bound on nl, eq.6}) is from (\ref{thm.bound on nl, eq.3}) and (\ref{thm.bound on nl, eq.4}).
Note that $T\bigcup S_1\bigcup S_2\subseteq E_{n,k}$. Then, $\NL_k(f)\geqslant\NL_{T\bigcup S_1\bigcup S_2}(f)$, and thus \[
\nl_{2i}^{(n)}\geqslant \frac{n}{2}{n/4-1\choose i-2}\left(2\nl_i^{(n/2)}-2{n/4-1\choose i-2}-1\right)+n{n/4-1\choose i-2}\nl_{i+1}^{(n/2)}.
\]
\end{proof}

\begin{example}\label{example. 16-variable WPB function}
    We use Construction-1 to design a $16$-variable WPB function. Suppose that we choose an $8$-variable WPB function $g$ as the subfunction of $f$ claimed in Construction-1, where $g$ achieves the best weightwise nonlinearity profile in Example~\ref{example.n=8}, i.e., $\NL_2(g)=9$, $\NL_3(g)=22$, $\NL_4(g)=27$. According to Theorem~\ref{thm.lower bound on NL recursively}, if we set $nl_2^{(8)}=9$, $nl_3^{(8)}=22$, and $nl_4^{(8)}=27$, then we obtain the lower bounds on $\NL_k(f)$, $3\leqslant k\leqslant 8$. See Table~\ref{weightwise non. pro. n=16}.

\begin{table}[t!]
\caption{\label{weightwise non. pro. n=16} Lower bound on weightwise nonlinearity profile of $f$ with $n=16$ variables}
   \begin{center}
    \begin{tabular}{|l|c|}%
     \hline
     $k$-weight nonlinearity of $f$ & upper bound $\left\lfloor {n \choose k}/2-\sqrt{ {n \choose k}}/2\right\rfloor$  \\  \hline
     ~~~~~~~~$\NL_2(f)\geqslant 5$ & $54$\\
     ~~~~~~~~$\NL_3(f)\geqslant 144$ & $268$\\
     ~~~~~~~~$\NL_4(f)\geqslant 472$ & $888$\\
     ~~~~~~~~$\NL_5(f)\geqslant 1056$ & $2150$\\
     ~~~~~~~~$\NL_6(f)\geqslant 2184$ & $3959$\\
     ~~~~~~~~$\NL_7(f)\geqslant 1296$ & $5666$\\
     ~~~~~~~~$\NL_8(f)\geqslant 2184$ & $6378$\\
     \hline
   \end{tabular}
   \end{center}
 \end{table}
\end{example}


\begin{remark}
    Grain-128~\cite{Hell2006} is a variant stream cipher selected in the eSTREAM project. It is shown in~\cite{Carlet2016} that the $17$-variable generate Boolean function $h'$ of Grain-128 is not WPB, and thus is  vulnerable to distinguish cryptanalysis when the attacker can access to the Hamming weight of the input of $h'$, especially for the weight larger than $8$. The weightwise nonlinearity profile of $h'$ is
    also studied in~\cite{Carlet2016}. Compared $h'$ with the $16$-variable WPB function $f$ in Example~\ref{example. 16-variable WPB function}, we conclude that
    \begin{itemize}
      \item for WPB property, $f$ provides the best resistance against distinguish attack,
      \item for $k$-weight nonlinearity, $f$ performs better than $h'$ if $k<5$, and may be worse otherwise.
    \end{itemize}


\end{remark}

\section{Concluding remarks}\label{Sec:Conclusion}

In this paper, we propose a large family of WPB functions over $\F_{2^n}$, where $n$ is a power of $2$. These WPB functions have algebraic degree $n-1$, and are EA inequivalent to the known constructions. By employing the Krawtchouk polynomial, we give a method to calculate the weightwise nonlinearity of these functions, and also prove that the $k$-weight nonlinearity of these functions are always nonzero when $k$ is a positive power of $2$. Moreover, we construct a subclass of WPB functions in our family, which have high weightwise nonlinearity profile.
This is the first time that a class of Boolean functions achieving the best possible balancedness and high nonlinearity simultaneously  with input restricted to constant weight vectors has been exhibited. Our work is beneficial in finding proper filter functions for special symmetric primitives like FLIP.



\hspace{2em}

\end{document}